\newif\ifFull
\Fullfalse

\ifFull
\documentclass[a4paper,11pt]{article}
\else
\documentclass[runningheads]{llncs}
\fi
\pdfoutput=1
\usepackage{graphicx}
\usepackage[english]{babel}
\usepackage{amssymb,amsfonts,amsmath}
\usepackage {color}
\usepackage[left,pagewise] {lineno}
\usepackage {xspace}
\usepackage {url}
\usepackage{subfig,wrapfig}
\usepackage[activate={true,nocompatibility},final,tracking=true,kerning=true,spacing=true,stretch=10,shrink=30]{microtype}
\usepackage[font=footnotesize,labelfont=bf]{caption}
\ifFull
\usepackage[margin=2.56cm]{geometry}
\fi

\newcommand{\niceremark}[3]{\textcolor{red}{\textsc{#1 #2: }}\textcolor{blue}{\textsf{#3}}}

\newcommand{\martin}[2][says]{\niceremark{Martin}{#1}{#2}}

\renewcommand{\niceremark}[3]{}

\renewcommand{\paragraph}[1]{\medskip\noindent\textbf{#1.}}

\ifFull
\newtheorem{theorem}{Theorem}
\newtheorem{corollary}{Corollary}

\newtheorem{definition}{Definition}

\newenvironment {proof}{\textbf {Proof:}}{\hfill \ensuremath {\boxtimes}}

\else
\let\doendproof\endproof
\renewcommand\endproof{~\hfill$\qed$\doendproof}
\fi

\definecolor {infocolor} {rgb} {0.6,0.6,0.6}

\title{On Minimizing Crossings\\ in Storyline Visualizations}
\ifFull
\author{
  Irina Kostitsyna\footnote{Dept. of Mathematics and Computer Science, TU Eindhoven, The Netherlands.}
  \and
  Martin N\"ollenburg\footnote{Algorithms and Complexity Group, TU Wien, Vienna, Austria.}
  \and
  Valentin Polishchuk\footnote{Link\"oping University, Link\"oping, Sweden.}
  \and
  Andr\'e Schulz\footnote{LG Theoretische Informatik, FernUniversit\"at in Hagen, Germany.}
  \and
  Darren Strash\footnote{Institute of Theoretical Informatics, Karlsruhe Institute of Technology, Germany.}
 }
\else
\institute{Dept. of Mathematics and Computer Science, TU Eindhoven, The Netherlands.
\and Algorithms and Complexity Group, TU Wien, Vienna, Austria.
\and Communications and Transport Systems, ITN, Link\"oping University, Sweden.
\and LG Theoretische Informatik, FernUniversit\"at in Hagen, Germany.
\and Institute of Theoretical Informatics, Karlsruhe Institute of Technology, Germany.}
\author {
    Irina Kostitsyna\inst{1}
  \and
  Martin N\"ollenburg\inst{2}
  \and
  Valentin Polishchuk\inst{3}
  \and
  Andr\'e Schulz\inst{4}
  \and
  Darren Strash\inst{5}
}
\fi

\begin{document}
\maketitle
\begin{abstract}
In a storyline visualization, we visualize a collection of interacting characters (e.g., in a movie, play, etc.) by $x$-monotone curves that converge for each interaction, and diverge otherwise. Given a storyline with $n$ characters, we show tight lower and upper bounds on the number of crossings required in any storyline visualization for a restricted case. In particular, we show that if (1) each meeting consists of exactly two characters and (2) the meetings can be modeled as a tree, then we can always find a storyline visualization with $O(n\log n)$ crossings. Furthermore, we show that there exist storylines in this restricted case that require $\Omega(n\log n)$ crossings. Lastly, we show that, in the general case, minimizing the number of crossings in a storyline visualization is fixed-parameter tractable, when parameterized on the number of characters $k$. 
Our algorithm runs in time $O(k!^2k\log k + k!^2m)$, where $m$ is the number of meetings.
\end{abstract}

\section {Introduction}
Ever since an xkcd comic\footnote{\url{http://xkcd.com/657}} featured storyline visualizations of various popular films, storyline visualizations have increasingly gained popularity as an area of research in the information visualization community (although the precursors of this kind of visualization may date back to Minard's 1861 visualization of Napoleon's Russian campaign of 1812). Informally, a storyline consists of characters (e.g., in a movie, play, etc.) who meet at certain times during a story. In a storyline visualization, each character is represented as an $x$-monotone curve. When characters meet (e.g., appear together in a scene, or interact), their representative curves should be grouped close together vertically, and otherwise their curves should be separate (see Fig.~\ref{figure:basic-storyline}, left). We assume that every character can only be in one meeting group at every point in time. One of the main goals for producing readable storyline visualizations is to minimize the number of crossings between character curves. Most previous results for constructing storyline visualizations are practical, implementing drawing routines that rely on heuristics or genetic algorithms~\cite{muelder-2013,tanahashi-2012}. However, there are only few theoretical results for storyline visualizations. Storyline visualization is tightly related to layered graph drawing~\cite{sugiyama-1981}, where layers correspond to meeting times in the storyline, and a permutation of all character curves needs to be computed for each time point.
Minimizing crossings in a storyline visualization is also related to bounding the ratio of (proper) crossings to touchings for families of monotone curves~\cite{pach-2015}.

\begin{figure}[t]
\centering
\includegraphics[width=.9\columnwidth]{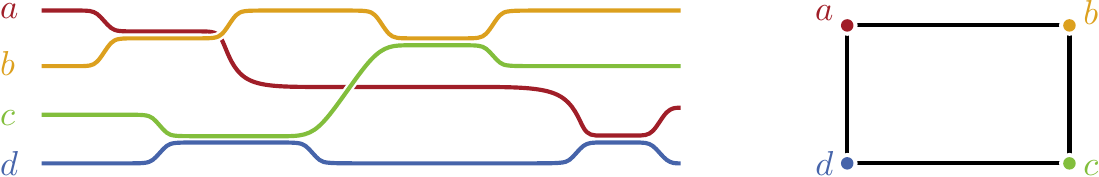}
\caption{Left: A storyline visualization with characters $a$, $b$, $c$, $d$. Right: The event graph.}
\label{figure:basic-storyline}
\end{figure}

\paragraph{Our Results}
While previous results focus on drawing storyline visualizations in practice using heuristics~\cite{muelder-2013,tanahashi-2012}, here we investigate the minimum number of crossings required in any storyline visualization. First, we investigate storyline visualizations in a restricted case. We show that if (1) each meeting consists of exactly two characters 
and (2) the meetings can be modeled as a tree,
then we can always find a storyline visualization with $O(n\log n)$ crossings, where $n$ is the number of characters.
Furthermore, we show that there exist storylines in this restricted case that require $\Omega(n\log n)$ crossings. Lastly, we show that, in the general case, minimizing the number of crossings in a storyline visualization is fixed-parameter tractable, when parameterized on the number of characters~$k$. Our algorithm runs in time $O(k!^2k\log k + k!^2m)$, where $m$ is the number of meetings.

\paragraph{Problem Formulation}
In the \emph{storyline problem}, we are given a storyline $S = (C, \mathcal{T}, \mathcal{E})$, that is defined by set of characters $C=\{1,\ldots,n\}$, that meet during closed time intervals $\mathcal{T} \subset \{[s,t]|s,t\in\mathbb{N}, s\leq t\}$. We call a meeting an \emph{event}, and denote the set of events as $\mathcal{E} \subset 2^{C} \times \mathcal{T}$, where each event $E_i = (C_i, [s_i,t_i])\in \mathcal{E}$ (with $1\leq i\leq m$) is defined by a subset $C_i \subseteq C$ of characters that meet for the entire time interval $[s_i, t_i] \in \mathcal{T}$ (naturally, a character cannot participate in two overlapping events). The goal then is to produce a 2D drawing of $S$, called a \emph{storyline visualization}, where the $x$-axis represents time, and characters are drawn as $x$-monotone curves 
placed in some vertical order for each point in time. During each event $E_i = (C_i, [s_i,t_i])$, curves representing characters in $C_i$ should be grouped within some small vertical distance $\delta_{\mathrm{group}}$ of each other, and otherwise the characters should be separated by some larger vertical distance ${\delta_{\mathrm{separate}}} > \delta_{\mathrm{group}}$. \martin{I would avoid the long indices and just say $\delta_\mathrm{g}$ and $\delta_\mathrm{sep}$}

\section{Pairwise Single-Meeting Storylines}
We focus on a simplified version of the storyline problem, where each event consists of exactly two characters, and these characters meet exactly once in $\mathcal{E}$.
\martin{reinserted single meeting. this is important here.}
For this simplified version, we can represent our events as a graph where every vertex is a character, and every edge is a meeting of the corresponding characters. We call this graph an \emph{event graph} (Fig.~\ref{figure:basic-storyline}, right).


\subsection{$O(n \log n)$ Crossings for Tree Event Graphs}
Let our event graph be a tree $T$ with $n$ nodes.
Then we show that we can always draw a storyline visualization with $O(n\log n)$ crossings. Our result relies on decomposing $T$ into disjoint subtrees that are drawn in disjoint axis-aligned rectangles. We reach this bound by using the \emph{heavy path decomposition} technique~\cite{sleator-tarjan-1983}.

\begin{definition}[heavy path decomposition~\cite{sleator-tarjan-1983}]
Let $T$ be a rooted tree. 
For each internal node $v$ in $T$, we choose a child $w$ with the largest subtree among all of $v$'s children. We call the edge $(v,w)$ a \emph{heavy edge}, and the edges to $v$'s other children \emph{light edges}. 
We call a maximal path of heavy edges a \emph{heavy path}, and the decomposition of $T$ into heavy paths and light edges a heavy path decomposition.
\end{definition}

\begin{figure}[!tb]
\centering
\includegraphics[width=0.47\textwidth]{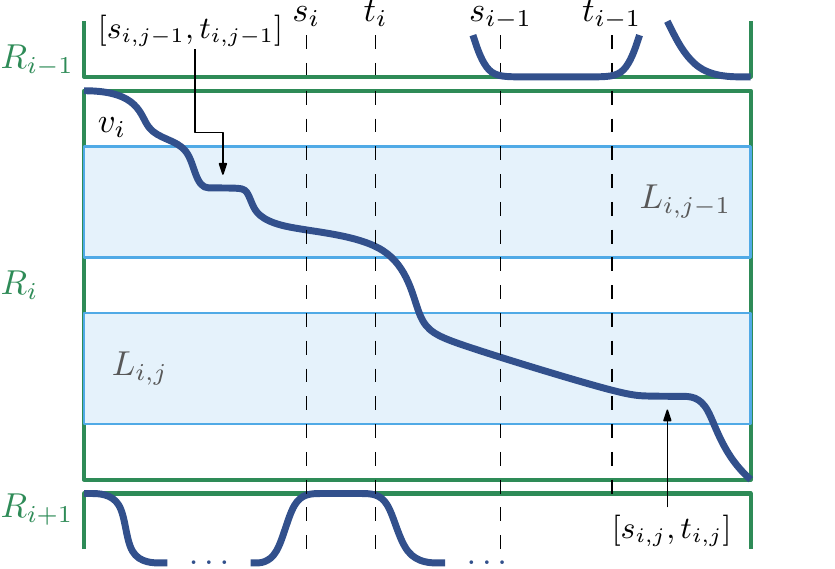}
\includegraphics[width=0.47\textwidth]{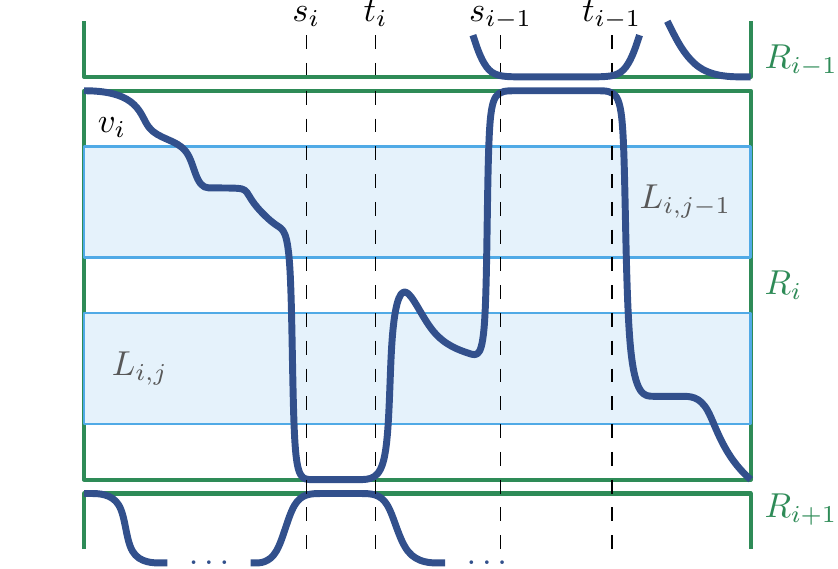}
\caption{The curve for $v_i$ before (left) and after (right) introducing detours.}
\label{figure:curve-detour}
\end{figure}

We first arbitrarily root $T$, and compute its heavy path decomposition.
Note that any root-leaf path of the event graph $T$ contains at most $\lceil \log n \rceil$ light edges~\cite{sleator-tarjan-1983}.  Let $P$ be the heavy path beginning at the root of $T$. We denote the node on $P$ at depth $i$ in $T$ by $v_i$. For each $v_i$, with $l_i$ light children, we first lay out each light subtree $L_{i,j}$ for $1 \leq j \leq l_i$.
We then order these layouts vertically in increasing order of meeting start time between $v_i$ and the root $r_{i,j}$ of $L_{i,j}$, separating each layout by vertical distance $\delta_{\mathrm{separate}}$. We denote the rectangle containing all layouts $L_{i,j}$ by $R_i$ (see Fig.~\ref{figure:curve-detour}). Then, we draw a single $x$-monotone curve from the top left to the bottom right of $R_i$, passing through the layout of each $L_{i,j}$, meeting the curve for each root $r_{i,j}$ at time $s_{i,j}$, and leaving at time $t_{i,j}$, for each event $(\{v_i, r_{i,j}\}, [s_{i,j},t_{i,j}])$.

Now for each $v_i$, we have a layout of $v_i$ and its light subtrees in a rectangle~$R_i$. We now show how to draw events between characters that are adjacent via a heavy edge in $P$. We first place all $R_i$ vertically in order along the path $P$ (from~$R_1$ to $R_{|P|}$), separated by distance $\delta_{\mathrm{group}}$.\martin{isn't this too close? it might create undesired meetings between top- and bottommost characters in two adjacent blocks.} We must have the curves meet for each event $(\{v_i, v_{i+1}\}, [s_i,t_i])$. We show how to introduce detours so that the curve $v_i$ joins curve $v_{i+1}$ at time $s_i$.
Let $n_i$ be the number of curves in the light subtrees of $v_i$.\martin{previously this was $l_i$ but that would be just the number of light children of $v_i$, not all characters in the light subtrees}
Before time $s_i$, curve $v_i$ has intersected some number $\gamma$ of the curves from its light subtrees, and has $n_i\!-\!\gamma$ curves still to intersect. Just before time $s_{i}$, we divert the curve so that it intersects the remaining $n_i\!-\!\gamma$ curves and reaches the bottom of rectangle $R_i$ to meet with $v_{i+1}$ at time $s_i$. Then at time $t_i$, we return the curve back to between curves $\gamma$ and $\gamma+1$ and allow the curve to continue as before, passing through the remaining $n_i\!-\!\gamma$ curves. For each $v_i$ we must also introduce a similar detour to the top of its rectangle $R_i$ so that it can meet the curve of $v_{i-1}$ at time $s_{i-1}$; see Fig.~\ref{figure:curve-detour}(right).

We introduce at most two such detours for each rectangle $R_i$, and therefore increase the number of crossings of each curve $v_i$ by a constant factor of at most five. \martin{is it obvious that $v_i$ does not cross any other curve more than 5 times? I think we can achieve it by making changes steep enough and always letting $v_i$ pause between two sublayouts and not within them. maybe this is something to say in a longer version...} Therefore, the total number of crossings $N(T)$ in our drawing of $T$ satisfies 
\ifFull
\else
$N(T) \leq \sum_{i =1}^{|P|} \sum_{j=1}^{l_i} N(L_{i,j}) + 5n$,
\fi
with base case $N((\{v\},\emptyset)) = 0$. Since all $L_{i,j}$ are disjoint, each iteration of the recurrence contributes at most $O(n)$ crossings. Further, since there are $O(\log n)$ light edges on the simple path from the root to any leaf in the heavy path decomposition~\cite{sleator-tarjan-1983}, the recurrence reaches the base case after $O(\log n)$ iterations. Therefore, the recurrence solves to $N(T) = O(n \log n)$ crossings\ifFull\else .\fi

\begin{theorem}
Any pairwise single-meeting storyline with a tree event graph has a storyline visualization with $O(n\log n)$ crossings.
\end{theorem}

\subsection{A Lower Bound}


\begin{wrapfigure}[8]{r}{.32\textwidth}
	\centering
	\vspace{-10ex}
	\includegraphics[scale=0.9]{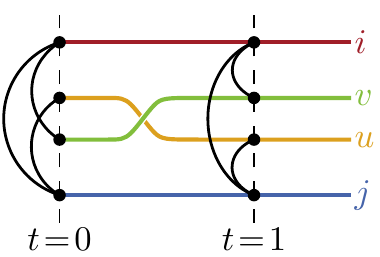}
	\caption{Event graph on line $t=0$ before and on line $t=1$ after swapping $u$ and $v$.}
	\label{figure:swaps}	
\end{wrapfigure}
Consider some storyline visualization $\mathcal{V}$ with an event graph $G$ with $n$ nodes and $m$ edges. Let $\pi_0$ be the ordering of the characters along a vertical line $t=0$ in~$\mathcal{V}$. Assign labels $[1,\dots,n]$ to the characters according to $\pi_0$. Then permutation $\pi_0$ defines an embedding of $G$ on the line $t=0$. As time progresses and character curves intersect, the corresponding vertices in the embedding of $G$ are swapped, see Fig.~\ref{figure:swaps}.

For every edge $e=(i,j)\in G$ define its cost $c_t(e)$ to be the number of characters between $i$ and $j$ on the vertical line at any given time $t$.

Then initially $c_0(e)=|i-j|-1$. 
So before $i$ and $j$ can meet, their curves must cross at least $|i-j|-1$ curves that were initially between
\ifFull
\else
them, which may be $0$.
\fi

When two character curves cross,
their corresponding vertices $u$ and $v$ swap in the embedding of $G$ on the vertical line. Notice that, after the swap, the costs of edges incident to $u$ or $v$ change by $\pm 1$, and there is no change for non-incident edges.
Thus, the crossing changes the cost of at most $\mathrm{deg}(u)+\mathrm{deg}(v)$ edges in $G$.

Let $C_0=\sum{c_0(e)}$ be the total initial cost of the edges of $G$ embedded on the line $t=0$. Then $C_0$ is the number of decrements in edge costs needed before all edges would have had cost $0$ at some moment in time. Every crossing of character curves $u$ and $v$ in $\mathcal{V}$ decreases this cost by at most $\mathrm{deg}(u)+\mathrm{deg}(v)$. Therefore, there are at least $\frac{\min_{\pi_0}C_0}{2\Delta}$ crossings in any storyline visualization $\mathcal{V}$ with an event graph $G$, where $\Delta$ is the maximum degree of $G$. Notice that $\min_{\pi_0}C_0=L^{*}-m$, where $L^{*}$ is the total edge length in the \emph{optimal linear ordering} of graph $G$ (the numbering of its vertices that minimizes the sum of differences of numbers over the graph's edges; see~\cite{linear-ordering} and~\cite[Problem GT42]{garey-johnson}).

\begin{theorem}
Any storyline visualization with an event graph $G$ requires $\Omega(\frac{L^{*}-m}{2\Delta})$ crossings, where $L^{*}$ is the total edge length of the optimal linear ordering of $G$, and $\Delta$ is the maximum degree of $G$.
\end{theorem}

\begin{corollary}
There exists a pairwise single-meeting storyline with a tree event graph whose storyline visualization requires $\Omega(n\log n)$ crossings.
\end{corollary}
\begin{proof}
Let $G$ be a full binary tree. Chung~\cite{chung-1978} showed that for any assignment of unique labels $[1,\ldots,n]$ to vertices of a full binary tree, the sum of label differences $|i-j|$ over all edges $(i,j)\in G$ is $\Omega(n\log n)$ (see also~\cite{tree-valuation}). 
Therefore, there will be $\Omega(\frac{\Omega(n\log n)-n+1}{2\times 3})=\Omega(n\log n)$ crossings.
\end{proof}

\section{An FPT Algorithm for the Storyline Problem}
We now consider general storylines, where any number of characters may participate in an event, and we have no restrictions on the meeting (hyper)-graph structure. The general storyline problem is NP-complete, by a straightforward reduction from  {\sc Bipartite Crossing Number}~\cite{gj83}.
However, in real-world storylines, there may be only a few characters of interest and these characters participate frequently in events. We therefore are interested in a parameterized algorithm to better capture the complexity in this scenario. Let $k = |C|$ be the number of characters in a storyline, and let $m = |\mathcal{E}|$ be the number of events. We show that the storyline problem is \emph{fixed-parameter tractable} when parameterized on~$k$. A problem is said to be fixed-parameter tractable if it can be solved in time $f(k)m^{O(1)}$, where $f$ is some function of $k$ that is independent of $m$.

\begin{theorem}
For storylines with $k$ characters and $m$ events, we can \ifFull\else solve the storyline problem \fi in time $O(k!^2k\log k + k!^2m)$.
\end{theorem}

\begin{proof}
We show how to reduce the storyline problem to finding shortest path in a graph. For each time interval $[s_i,t_i]$ in the storyline we take its start time $s_i$ and create a vertex for each of the $O(k!)$ possible vertical orderings of the curves that satisfy the event groupings at $s_i$. We denote the vertices for time $s_i$ by $v_{i,j}$, where $1 \leq j \leq k!$, and say these vertices are on \emph{level} $i$.

Denote the minimum number of crossings to transform one ordering $v_{i,j}$  at level $i$ to ordering $v_{i+1,l}$ at level $i+1$ by $I(v_{i,j},v_{i+1,l})$. For all levels, we connect each vertex $v_{i,j}$ to each vertex $v_{i+1,l}$ by a directed edge with weight $I(v_{i,j},v_{i+1,l})$.
\ifFull
\else
We then create source and terminal vertices $s$ and $t$ and connect them with edges of weight $0$ to vertices on levels $1$ and $m$, respectively.
\fi
Then the weight of a shortest path from $s$ to $t$ is the minimum number of crossings in any embedding, and this path specifies the vertical orderings of the curves at each time step~$s_i$.

We now compute the number of crossings to transform between vertical orderings. First note that we can compute the minimum number of swaps between two vertical orderings of size $k$ in time $O(k\log k)$
\ifFull
\else
by counting inversions with merge sort.
\fi
Thus, we can precompute the weights between all pairs of orderings in time $O(k!^2k\log k)$, and assign edge weights when building the graph at a cost of $O(k!^2)$ per level.

Now a minimum-weight path from $s$ to $t$ fully specifies a storyline visualization. We can lay out each curve by the vertical ordering specified by each vertex on the path with its time step, swapping curve order between time steps. Then during each event we group the curves together, otherwise we separate them.

In total there are $m$ levels, each with $O(k!)$ vertices and $O(k!^2)$ edges. Thus, there are $O(k!m)$ vertices and $O(k!^2m)$ edges. We can compute a shortest path from $s$ to $t$ in time linear in the number of vertices and edges, by dynamic programming: For each level $i$, we compute the minimum weight for each vertex $v$ by iterating over all incoming edges from vertices on level $i-1$ and choosing the one that minimizes the total weight to $v$. Thus we can compute a shortest path from $s$ to $t$ in time $O(k!^2m)$. Including the time to precompute edge weights, we get total time $O(k!^2k\log k) + O(k!^2m) = O(k!^2k\log k + k!^2m)$.
\end{proof}


\subsubsection{Acknowledgments.}We thank the anonymous referees for their helpful comments. This research was initiated at the 2nd International Workshop on Drawing Algorithms for Networks in Changing Environments (DANCE 2015) in Langbroek, the Netherlands, supported by the Netherlands Organisation for Scientific Research (NWO) under project no. 639.023.208. IK is supported in part by the NWO under project no. 639.023.208. VP is supported by grant 2014-03476 from the Sweden's innovation agency VINNOVA.


\end{document}